\documentclass[journal]{IEEEtran}
\usepackage{amssymb}
\usepackage{amsmath}
\usepackage{longtable}
\newtheorem{theorem}{Theorem}[section]
\newtheorem{lemma}[theorem]{Lemma}

\newtheorem{corollary}[theorem]{Corollary}

\newenvironment{proof}[1][Proof]{\begin{trivlist}
\item[\hskip \labelsep {\bfseries #1}]}{\end{trivlist}}

\newenvironment{example}[1][Example]{\begin{trivlist}
\item[\hskip \labelsep {\bfseries #1}]}{\end{trivlist}}

\ifCLASSINFOpdf
\else
\fi
\begin{document}
%
\title{A Note on Weight Distributions
of Irreducible Cyclic Codes}
%
%
%

\author{ Chunming~Tang,
 Yanfeng~Qi, Maozhi~Xu, Baocheng~Wang, Yixian~Yang
\thanks{C. Tang, Y. Qi and M. Xu is with 
Laboratory of Mathematics and Applied Mathematics, School of Mathematical Sciences
, Peking University,  100871, China

C. Tang's e-mail: tangchunmingmath@163.com}
\thanks{B. Wang and Y, Yang are with
Department of Computer, Beijing University of
Posts and Telecommunications, Beijing, 100876, China}
}


%
%

\markboth{ }%
{Shell \MakeLowercase{\textit{et al.}}: A Note on Weight Distributions
of Irreducible Cyclic Codes}
%



\maketitle

\begin{abstract}
Usually, it is difficult to determine the weight distribution of an irreducible cyclic code. In this paper,
we
discuss the case when  an irreducible cyclic code has the maximal number of distinct nonzero weights and give a necessary and sufficient condition. In this case, we also obtain a divisible property for the weight of a codeword.  Further, we present a necessary and sufficient condition for an irreducible cyclic code with only one nonzero weight.  Finally, we determine the weight distribution of an irreducible cyclic code for some cases.
\end{abstract}

\begin{IEEEkeywords}
Cyclotomic fields, irreducible cyclic codes,  Gauss sums,  Gaussian periods,  weight distributions.
\end{IEEEkeywords}

%
\IEEEpeerreviewmaketitle

\section{Introduction}
%
%
%
%

Let $p$ be a prime, $q=p^s$ and
$r=q^m$, $k=sm$, where $s$ and $m$ are two positive integers. Suppose that an integer $N$ satisfies
$N|(r-1)$. Let $n=\frac{r-1}{N}$ and
$N_{2}=(N,\frac{r-1}{q-1})$.
Given two  integers $a$ and $b$, which are coprime, $ord_{a}b$ denotes the order of $a$ module $b$. Consider
the finite field $GF(r)^*=<\alpha>$.
Let $\theta=\alpha^{N}$. The set
\begin{align*}
&\mathcal{C}(r,N)=\\
&\{
(Tr_{r/q}(\beta), Tr_{r/q}(\beta\theta),
\cdots, Tr_{r/q}(\beta\theta^{n-1})): \beta
\in GF(r)
\}
\end{align*}
is called an irreducible cyclic $[n,m_{0}]$ code, where
$Tr_{r/q}$ is the trace function from $GF(r)$ onto
$GF(q)$  and $m_{0}=dim_{GF(q)}(GF(q)(\theta))$
\cite{DH}.
Then $m_{0}|m$. Generally, $m_{0}=m$  is satisfied. Hence, we just consider the case $m_{0}=m$ in this paper.

The weight distribution of $\mathcal{C}(r,N)$ attracts much interest. Generally, it is difficult to determine the weight distribution\cite{MS}.
For any $\beta\in GF(r)^{*}$, the Hamming weight of any codeword
$$
c(\beta)=(Tr_{r/q}(\beta), Tr_{r/q}(\beta\theta),
\ldots, Tr_{r/q}(\beta\theta^{n-1}))
$$
in the code $\mathcal{C}(r,N)$ can be represented by a linear combination of Gauss sums via Fourier transform\cite{Mc1,Mc2,MR}. Hence, Gauss sums can be used to analyze the weight distribution of the irreducible cyclic codes. Some main results on the weight distribution are list here.

{\rm $\bullet$} Let $N|(q^j+1)$, where
$j|\frac{m}{2}$. This is called the semi-primitive case \cite{BM1,DG,Mc1}.
Then $\mathcal{C}(r,N)$ is a two weight code.

{\rm $\bullet$} Let $N$ be a prime satisfying
$N\equiv 3 \pmod 4$ and $ord_{N}(q)=(N-1)/2$.
Baumert and Mykkeltveit\cite{BM2} determined the corresponding weight distribution.

{\rm $\bullet$} Let
$N=2$. Baumert and McEliece\cite{BM1} gave the corresponding
weight distribution.

{\rm $\bullet$} Let $N=3,4$, Ding \cite{Ding2} determined
the weight distribution.

{\rm $\bullet$} Let $N_2=1,2$, Ding \cite{Ding2} determined
the weight distribution.

Schimidt and White \cite{SW} got necessary and sufficient  conditions for an irreducible cyclic code having at most two weights. This paper considers two other cases. One is that
a irreducible cyclic code has the maximal number of distinct nonzero weights;
the other is that a irreducible cyclic code has only one nonzero weight.
More results on the weight distribution can be found in \cite{AL,HKM,La,MV,SW,VV}.

In this paper, we simplify the formula of
$wt(c(\beta))$ \cite{Ding2}, which is represented by
Gauss periods. From the simplified formula, we find that determining  the weight distribution of $\mathcal{C}(r,N)$
is equivalent to determining the weight distribution of $\mathcal{C}(r,N_{2})$. Further, determining  the weight distribution of $\mathcal{C}(r,N)$ is reduced to the factorization of $\psi^*_{(N_{2},r)}(X)$.
From cyclotomic fields and Stickelberger theorem on Gauss sums,
We present a necessary and sufficient condition for an irreducible cyclic code having the maximal number of different nonzero weights and obtain a divisible property of the weight of a codeword.
Then we generalize the results of Ding \cite{Ding2}  to cases $N_2=3,4$ and give  the necessary and sufficient conditions for a code with only one nonzero weight. Finally, we present the weight distributions of the irreducible cyclic codes for
cases $N=5,6,8,12$.

\section{Prelimilaries}
For determining the weight distribution, we first
recall cyclotomic classes,  Gauss sums and reduced period polynomials,

The definition of cyclotomic classes of order
$N$ in $GF(r)$ is the coset
$$
C_{i}^{(N,r)}=\alpha^i<\alpha^N>.
$$
When $i\equiv j\pmod N$, $C_{i}^{(N,r)}=
C_{j}^{(N,r)}$. If $\beta \in C_{i}^{(N,r)}$,
we let $i(\beta)=i$.

The Gauss periods are
$$
\eta_{i}^{(N,r)}=\sum_{x\in C_{i}^{(N,r)}}
\mu(x),
$$
where $\mu(x)=\zeta_{p}^{Tr_{r/p}(x)}$,
$\zeta_{p}=exp(2\pi i/p)$ and $i=\sqrt{-1}$.

Let $S$ be a subset of $GF(r)$, $\mu(S)$ denotes
$\sum_{x\in S}\mu(x)$. Then $\eta_{i}^{(N,r)}=
\mu(C_{i}^{(N,r)})$. Let $\eta_{i}^{*(N,r)}=
1+N\eta_{i}^{(N,r)}$, then $\eta_{i}^{*(N,r)}$ is called the reduced period. The reduced  periods polynomial is of the form
$$
\psi_{(N,r)}^{*}(X)=
(X-\eta_{0}^{*(N,r)})\cdots(X-\eta_{N-1}^{*(N,r)}).
$$

\subsection{Gauss sums}
In this subsection, we introduce some knowledge on Gauss sums \cite{LH}.

Let $\chi: GF(r)^*\rightarrow C^{*}$ be a character of $GF(r)^*$, where $C$ is the complex field. Then a Gauss sum is defined by
$$
g(\chi)=-\sum_{\beta\in GF(r)}\chi(\beta)\mu(\beta).
$$
For any $c_{2}(\beta)\in \mathcal{C}(r,N_{2})$, we have the following McEliece's identity \cite{Mc2}
\begin{eqnarray}\label{for1}
wt(c_{2}(\beta))=\frac{q-1}{N_{2}q}(r+\sum_{\chi^{
N_{2}}=1,\chi\neq 1}g(\chi)\overline{\chi}(\beta)).
\end{eqnarray}
To determine $wt(c_{2}(\beta))$, we can utilize properties of the Gauss sum $g(\chi)$.

Given an integer $t$, let $\zeta_{t}=exp(2\pi i/t)$. Let $\wp$ be a prime ideal of $Q(\zeta_{r-1})$ over $p$ and $\widetilde{\wp}$ is a prime ideal of
$Q(\zeta_{r-1},\zeta_{p})$ over $\wp$.
An integer $h \;(0\leq h<r-1)$ can be represented by $h=h_{0}+h_{1}p+\cdots+h_{k-1}p^{k-1}$.
Then we let $s(h)=\sum_{i=0}^{k-1}h_{i}p^i$.

There is an isomorphism
$$
Z[\zeta_{r-1}]/\wp\longrightarrow GF(r).
$$
Note that all the $r-1$-th root of unity are different module $\wp$. Then we have the following isomorphism \cite{Wa}
$$
\omega(\cdot): GF(r)^*\longrightarrow \{
\zeta_{r-1}^{0},\zeta_{r-1}^{1},\cdots,\zeta_{r-1}^{
r-2}
\},
$$
which satisfies
$$
\omega(\beta)\mod \; \wp=\beta\in GF(r)^*.
$$
From Stickelberger's theorem \cite{St}, we have
\begin{eqnarray}\label{for2}
v_{\widetilde{\wp}}(g(\omega^{-h}))=s(h).
\end{eqnarray}
Further, we have the following relation from Lang
\cite{La1,La2}
\begin{eqnarray}\label{for3}
g(\omega^{-h})\equiv \frac{(\zeta_{p}-1)^{s(h)}}
{(h_{0}!)\cdots(h_{k-1}!)} \pmod {
\widetilde{\wp}^{s(h)+1}}.
\end{eqnarray}

Some results on the factorization of some reduced period polynomials are listed here.
\subsection{The factorization of $\psi_{(N,r)}^{*}$ for cases $N=2,3,4$ \cite{My}}
\begin{lemma}\label{lem1}
Let $~2|\frac{r-1}{p-1}$. Then we have the following results on
the factorization of $~\psi_{(2,r)}^{\ast}(X)$
$$\psi_{(2,r)}^{\ast}(X)=(X+\sqrt{r})(X-\sqrt{r}).$$
\end{lemma}

\begin{lemma}\label{lem2}
Let $~3|\frac{r-1}{p-1}.$ We have the following results on
the factorization of $~\psi_{(3,r)}^{\ast}(X).$\\
(a) If $p\equiv 1 ~(mod ~3),$ then $~3|k,$ and
\begin{align*}
\psi_{(3,r)}^{\ast}(X)=(X-cr^{1/3})
(X+\frac{1}{2}(c+9d)r^{1/3})\\
\times(X+\frac{1}{2}(c-9d)r^{1/3})
\end{align*}
where $~c$ and $~d$ are given by $4p^{k/3}=c^2+27d^2,~c\equiv 1~(mod
~3),$ and $~gcd(c,p)=1.$\\
(b) If $~p\equiv 2~(mod ~3),$ then $~2|k$, and
\begin{align*}
\psi_{(3,r)}^{\ast}(X)=(X+(-1)^{k/2}2\sqrt{r})
(X-(-1)^{k/2}\sqrt{r})^2.
\end{align*}
\end{lemma}

\begin{lemma}\label{lem3}
Let $~4|\frac{r-1}{p-1}.$ We have the following results on
the factorization of $~\psi_{(4,r)}^{\ast}(X).$\\
(a) If $~p\equiv 1 ~(mod ~4),$ then $~4|k,$ and
\begin{align*}
\psi_{(4,r)}^{\ast}(X)=(X+\sqrt{r}+2r^{1/4}u)(X+\sqrt{r}-2r^{1/4}u)\times\\
(X-\sqrt{r}+4r^{1/4}v)(X-\sqrt{r}-4r^{1/4}v)
\end{align*}
where $~u$ and $~v$ are given by $~p^{k/2}=u^2+4v^2,~u\equiv
1~(mod~4)$ and $~gcd(u,p)=1.$\\
(b) If $~p\equiv 3 ~(mod ~4),$ then $~2|k,$ and
\begin{align*}
\psi_{(4,r)}^{\ast}(X)=(X+(-1)^{k/2}3\sqrt{r})(X-(-1)^{k/2}\sqrt{r})^3.
\end{align*}
\end{lemma}
\subsection{The factorization of $\psi_{(5,r)}^{*}$}
In order to describe the explicit factorization of the quintic period polynomials for finite fields, we require to discuss integer solutions of Dickson's system.
\[\left\{
  \begin{array}{l}
    16\sqrt[5]{r}=x^2+125w+50v^2+50u^2, \\
    xw=v^2-4vu-u^2, \\
    x\equiv -1 \pmod 5.
  \end{array}
\right.\]
Let $\sigma$ be a non-singular linear transformation of order 4.
\begin{eqnarray*}
\sigma(x,w,v,u)=(x,-w,-u,v).
\end{eqnarray*}
Then we have the following lemma for the explicit factorization of period polynomials \cite{Ho}.
\begin{lemma}\label{lem4}
Let $5|\frac{r-1}{p-1}$. We have the following results on the factorization of $\psi_{(5,r)}^{*}(X)$.

{\rm (a)} If $p\equiv 1 \pmod 5$, then $5|k$ and
$$
\psi_{(5,r)}^{*}(X)=(X+\frac{\sqrt[5]{r}}{16}(x^3
-25L))\prod_{i=0}^{3}(X-\frac{\sqrt[5]{r}}{64}
\sigma^i(X^3-25M)),
$$
where $(x,w,v,u)$ is the integer solution of Dickson's system satisfying the condition
$p\nmid (x^2-125w^2)$ and
\begin{eqnarray*}
L&=&2x(v^2+u^2)+5w(11v^2-4vu-11u^2),\\
M&=&2x^2u+7xv^2+20xvu-3xu^2+125w^3\\
&&+200w^2v-150w^2u+5wv^2
-20wvu\\
&&-105wu^2-40v^3-60v^2u+120vu^2+20u^3.
\end{eqnarray*}

{\rm (b)} If $p\not\equiv 1\pmod 5$, then
$$
\psi_{(5,r)}^{*}(X)=\left\{
  \begin{array}{ll}
    (X-r^{1/2})^4(X+4r^{1/2}), \text{if}\; k/ord_5p\; \text{is even}\\
    (X+r^{1/2})^4(X-4r^{1/2}), \text{if}\;
k/ord_5p\; \text{is odd}
  \end{array}
\right.$$
\end{lemma}
\subsection{The factorization of $\psi_{(N,r)}^{*}$ for cases $N=6,8,12$}
To determine the factorization of the reduced period polynomial $\psi_{(N,r)}^*$ for cases $N=6,8,12$ \cite{Gu}, we need to introduce some notations.

Let $G:=\alpha^{\frac{r-1}{q-1}}$ and $Z$ be the least positive integer satisfying $G^{Z}=2$.

When $p\equiv 1\pmod 3$. Then there
exist two unique integers $r_{3}$ and
$s_{3}$ satisfying
\begin{eqnarray*}
4p&=&r_{3}^{2}+3s_{3}^2\\
r_{3}&\equiv& 1\pmod 3\\
s_{3}&\equiv& 0\pmod 3\\
3s_{3}&\equiv& (2G^{(p-1)/3}+1)r_{3} \pmod p
\end{eqnarray*}
Let $\lambda=(r_{3}+i\sqrt{3}s_{3})/2$ and a sequence related to $\lambda$ be
$$
V_{j,n}=\zeta_{6}^{-j}\lambda^n+\zeta_{6}^j
\overline{\lambda}^n.
$$
When the prime $p\equiv 1\pmod 4$, there exist
two unique integers $a_{4}$ and $b_{4}$ satisfying
\begin{eqnarray*}
p&=&a_{4}^2+b_{4}^2\\
a_{4}&\equiv& -(-1)^Z \pmod 4\\
b_{4}&\equiv & a_{4}G^{(p-1)/4}\pmod p.
\end{eqnarray*}
Then we can define $\pi=a_{4}+ib_{4}$ and
sequences related to $\pi$.
\begin{eqnarray*}
Q_{n}&=\pi^n+\overline{\pi}^n, &
P_{n}=-i(\pi^n-\overline{\pi}^n),\\
Q_{j,n}&=\zeta_{4}^{-j}\pi^n+\zeta_{4}^j
\overline{\pi}^n, &
P_{j,n}=-i(\zeta_{4}^{-j}\pi^n-\zeta_{4}^j
\overline{\pi}^n).
\end{eqnarray*}
When $p\equiv 1\pmod 8$, there exist two unique
integers $a_{8}$ and $b_{8}$ satisfying
\begin{eqnarray*}
p&=&a_{8}^2+2b_{8}^2\\
a_{8}&\equiv & -1 \pmod 4\\
2b_{8}&\equiv & (G^{(p-1)/8}+G^{3(p-1)/8})a_{8}\pmod p.
\end{eqnarray*}
When $p\equiv 3\pmod 8$, there exist two unique
integers $a_{8}$ and $b_{8}$ satisfying
\begin{eqnarray*}
p&=&a_{8}^2+2b_{8}^2\\
a_{8}&\equiv & (-1)^{(p-3)/8} \pmod 4\\
2b_{8}&\equiv & (\alpha^{(q-1)/8}-\alpha^{(1-q)/8})a_{8}\pmod p.
\end{eqnarray*}
For the two cases, we define
$\sigma=a_{8}+ib_{8}\sqrt{2}$ and sequences related to $\sigma$.
$$
T_{n}=\sigma^n+\overline{\sigma}^n, \quad
S_{n}=(\sigma^n-\overline{\sigma}^n)/(i\sqrt{2}).
$$

\begin{lemma}\label{lem5}
Let $6|\frac{r-1}{p-1}$. We have the following results on the factorization of $\psi_{(6,r)}^{*}(X)$.

{\rm (a)} If $p\equiv 1 \pmod 6$, then $6|k$ and
$$
\psi_{(6,r)}^{*}(X)=
(X-\eta_{1}^{*(6,r)})\cdots (X-\eta_{6}^{*(6,r)}),
$$
where $
\eta_{j}^{*(6,r)}=-(-1)^{tk/2}p^{k/6}
V_{j,2k/3}-p^{k/3}V_{2j,k/3}-
(-1)^{j+tk/2}p^{k/2}$, $t=\frac{p-1}{6}$  and $V_{j,n}$ are defined above.

{\rm (b)} If $p\equiv 5\pmod 6$, then
$$
\psi_{(6,r)}^{*}(X)=
(X-(-1)^{k/2}p^{k/2})^5(X+5(-1)^{k/2}p^{k/2}).$$
\end{lemma}

\begin{lemma}\label{lem6}
Let $8|\frac{r-1}{p-1}$. We have the following results on the factorization of $\psi_{(8,r)}^{*}(X)$.

{\rm (a)} If $p\equiv 1 \pmod 8$, then $8|k$ and
$$
\psi_{(8,r)}^{*}(X)=
(X-\eta_{1}^{*(8,r)})\cdots (X-\eta_{8}^{*(8,r)}),
$$
where $\eta_{j}^{*(8,r)}=-p^{k/2}-
p^{k/4}Q_{j,k/2}-p^{k/8}AB$ and
$A, B$ are defined as follows:
If $j$ is even, $A=Q_{j/2,k/4}$ and $B=T_{k/2}$. If $j$ is odd, $A=(-1)^{[{j}/{4}]}Q_{0,k/4}+
        (-1)^{[{(j-2)}/{4}]}P_{0,k/4}$ and $B=S_{k/2}$.

{\rm (b)} If $p\equiv 3\pmod 8$, then
$4|k$ and
$$
\psi_{(8,r)}^{*}(X)=
(X-\xi_{1})^2(X-\xi_{2})^2(X-\xi_{3})^2
(X-\xi_{4})(X-\xi_{5}),
$$
where
$\xi_{1}=-2p^{k/2}S_{k/2}+p^{k/2}$,  $\xi_{2}=2p^{k/2}S_{k/2}+p^{k/2}$,
$\xi_{3}=p^{k/2}$, $\xi_{4}=2p^{k/4}T_{k/2}-3p^{k/2}$ and
$\xi_{5}=-2p^{k/4}T_{k/2}-3p^{k/2}.$

{\rm (c)} If $p\equiv 5\pmod 8$, then $8|k$ and
\begin{align*}
\psi_{(8,r)}^{*}(X)=(X-\xi_{1})^2(X-\xi_{2})^2
(X-\xi_{3})(X-\xi_{4})\\
\times(X-\xi_{5})(X-\xi_{6})
,
\end{align*}
where
$\xi_{1}=p^{k/2}-p^{k/4}P_{k/2}$, $\xi_{2}=p^{k/2}+p^{k/4}P_{k/2}$,
$\xi_{3}=-p^{k/2}-p^{k/4}Q_{k/2}-2p^{3k/8}Q_{k/4}$, $\xi_{4}=-p^{k/2}-p^{k/4}Q_{k/2}+2p^{3k/8}Q_{k/4}$, $\xi_{5}=-p^{k/2}+p^{k/4}Q_{k/2}-2p^{3k/8}P_{k/4}$ and $\xi_{6}=-p^{k/2}+p^{k/4}Q_{k/2}+2p^{3k/8}P_{k/4}$.

{\rm (d)} If $p\equiv 7\pmod 8$, then $2|k$ and
$$
\psi_{(8,r)}^{*}(X)=(X-(-1)^{k/2}p^{k/2})^7
(X+7(-1)^{k/2}p^{k/2})
.$$
\end{lemma}

\begin{lemma}\label{lem7}
Let $12|\frac{r-1}{p-1}$. We have the following results on the factorization of $\psi_{(12,r)}^{*}(X)$.

{\rm (a)} If $p\equiv 1 \pmod {12}$, then $12|k$ and
$$
\psi_{(12,r)}^{*}(X)=
(X-\eta_{1}^{*(12,r)})\cdots (X-\eta_{1}^{*(12,r)}),
$$
where $\eta_{j}^{*(12,r)}=-p^{k/12}Q_{j,k/2}V_{-j,k/3}
-p^{k/4}Q_{j,k/2}
-p^{k/6}V_{j,2k/3}-p^{k/3}V_{2j,k/3}
-(-1)^jp^{k/2}$.

{\rm (b)} If $p\equiv 5\pmod {12}$, then
$4|k$ and
\begin{align*}
\psi_{(12,r)}^{*}(X)=
(X-\xi_{1})^2(X-\xi_{2})^2(X-\xi_{3})^2
(X-\xi_{4})^2 \\
\times
(X-\xi_{5})\cdots
(X-\xi_{8}),
\end{align*}
where
$\xi_{1}=Q_{k/2}p^{k/4}((-1)^{k/4}-1)+p^{k/2}$, $\xi_{2}=-Q_{k/2}p^{k/4}((-1)^{k/4}-1)+p^{k/2}$,
$\xi_{3}=P_{k/2}p^{k/4}((-1)^{k/4}+1)+p^{k/2}$, $\xi_{4}=-P_{k/2}p^{k/4}((-1)^{k/4}+1)+p^{k/2}$,
 $\xi_{5}=P_{k/2}p^{k/4}(2(-1)^{k/4}-1)+p^{k/2}$,  $\xi_{6}=-P_{k/2}p^{k/4}(2(-1)^{k/4}-1)+p^{k/2}$,
$\xi_{7}=Q_{k/2}p^{k/4}(2(-1)^{k/4}+1)-5p^{k/2}$ and  $\xi_{8}=-Q_{k/2}p^{k/4}(2(-1)^{k/4}-1)-5p^{k/2}$.

{\rm (c)}  If $p\equiv 7\pmod {12}$. Let $\rho=2(-1)^{k(p+5)/6}$, then $6|k$ and
$$
\psi_{(12,r)}^{*}(X)=
(X-\eta_{1}^{*(12,r)})\cdots (X-\eta_{j}^{*(12,r)})
,$$
where $\eta_{j}^{*(12,r)}$ is defined as follows:
If $j$ is odd, $\eta_{j}^{*(12,r)}=-(-1)^{k/2}p^{k/6}V_{j,2k/3}
-p^{k/3}V_{2j,k/3}+(-1)^{k/2}p^{k/2}$; if
$2\|j$, $\eta_{j}^{*(12,r)}=-(-1)^{k/2}p^{k/6}V_{j,2k/3}
+p^{k/3}V_{2j,k/3}(\rho-1)+p^{k/2}(\rho-(-1)^{k/2})$; if $4\mid j$, $\eta_{j}^{*(12,r)}=-(-1)^{k/2}p^{k/6}V_{j,2k/3}
-p^{k/3}(\rho+1)V_{2j,k/3}-p^{k/2}(\rho+(-1)^{k/2})$.

{\rm (d)} If $p\equiv 11\pmod {12}$, then $2|k$ and
$$
\psi_{(12,r)}^{*}(X)=(X-(-1)^{k/2}p^{k/2})^{11}
(X+11(-1)^{k/2}p^{k/2})
.$$
\end{lemma}

\section{Some general results on the weight distribution}
To discuss the weight distribution, we first consider
$$
Z(r,\beta)=\#\{Tr_{r/q}(\beta x^N)=0:
x\in GF(r)\}.
$$
Then $wt(c(\beta))=n-\frac{Z(r,\beta)-1}{N}$.

For $Z(r,\beta)$, we have the following formula \cite{Ding2}.
\begin{eqnarray*}
Z(r,\beta)&=&\frac{1}{q}\sum_{y\in GF(q)}\sum_{
x\in GF(r)}\zeta_p^{Tr_{q/p}(yTr_{r/q}(\beta x^N))}\\
&=&\frac{1}{q}[q+r-1+N\sum_{y\in GF(q)^*}
\sum_{x\in \mathcal{C}_{0}^{N,r}}\mu(y\beta x)].
\end{eqnarray*}

For the simplification of this formula, we introduce some lemmas.
\begin{lemma}\label{lem3.1}
We use the assumptions and notations above, then

{\rm (a)} $\#(GF(q)^{*}\cap{C}_{0}^{(N,r)})
=
\frac{r-1}{[N,\frac{r-1}{q-1}]}$.

{\rm (b)} $GF(q)^{*}\cap{C}_{i}^{(N,r)}
\neq\emptyset$ if and only if $(N,\frac{r-1}{q-1})|i$.

{\rm (c)} If $GF(q)^{*}\cap{C}_{i}^{(N,r)}
\neq\emptyset$, then $\#(GF(q)^{*}\cap{C}_{i}^{(N,r)})
=
\frac{r-1}{[N,\frac{r-1}{q-1}]}$.
\end{lemma}
\begin{proof}
{\rm (a)} Since
$GF(q)^{*}$ and ${C}_{0}^{(N,r)}$ are subgroups of the cyclic group $GF(r)^{*}$, where
$$GF(q)^{*}=<\alpha^{\frac{r-1}{q-1}}> ~~and~~
\mathcal{C}_{0}^{(N,r)}=
<\alpha^{N}>,$$
then
$$GF(q)^{*}\cap {C}_{0}^{(N,r)}=<\alpha^{
[\frac{r-1}{q-1},N]}>.$$
that is
$$\#(GF(q)^{*}\cap {C}_{0}^{(N,r)})
=\frac{r-1}{[N,\frac{r-1}{q-1}]}.$$

{\rm (b)} If $(N,\frac{r-1}{q-1})|i$, then there exists two integers $a$ and $b$ satisfying
$$Na+\frac{r-1}{q-1}b=i.$$
 Hence
$$\alpha^{\frac{r-1}{q-1}b}=\alpha^{-Na}\cdot
\alpha^i\in (GF(q)^{*}\cap {C}_{i}^{(N,r)}).$$
that is
$$GF(q)^{*}\cap{C}_{i}^{(N,r)}
\neq\emptyset.$$
If $GF(q)^{*}\cap{C}_{i}^{(N,r)}
\neq\emptyset$, there exists two integers $a$ and
$b$ satisfying
$$\alpha^{\frac{r-1}{q-1}a}
=\alpha^{Nb}\cdot \alpha^{i}.$$
that is,
$$\alpha^i=\alpha^{\frac{r-1}{q-1}a}\cdot \alpha^{-Nb}
\in <\alpha^{\frac{r-1}{q-1}},\alpha^N>.$$
Note that
$$<\alpha^{\frac{r-1}{q-1}},\alpha^N>=
<\alpha^{(\frac{r-1}{q-1},N)}>.$$
Hence
$$\alpha^i\in <\alpha^{(\frac{r-1}{q-1},N)}>.$$
Then we have
$$(\frac{r-1}{q-1},N)|i.$$

{\rm (c)} If $GF(q)^{*}\cap{C}_{i}^{(N,r)}
\neq\emptyset$, we can let $GF(q)^{*}\cap{C}_{i}^{(N,r)}=n_{i}$.

We first prove that $n_{i}\geq n_0$.
All the elements in $GF(q)^{*}\cap{C}_{0}^{(N,r)}$ are denoted by
$$x_{1}=y_{1},$$
$$\vdots~~~~\vdots$$
$$x_{n_{0}}=y_{n_{0}},$$
where
$x_{1},\cdots,x_{n_{0}}\in GF(q)^{*}$
and $y_{1},\cdots,y_{n_{0}}\in {C}_{0}^{(N,r)}$.
Take an element $z_1=w_1$ in
$GF(q)^{*}\cap{C}_{i}^{(N,r)}$, where
$z_1\in GF(q)^{*}$ and $w_1\in {C}_{i}^{(N,r)}$. Then
$z_1x_j\in GF(q)^{*}$ and $w_1y_j\in {C}_{i}^{(N,r)}$. Hence
$$z_1x_j=w_1y_j\in GF(q)^{*}\cap{C}_{i}^{(N,r)}\;(1\leq j\leq n_0).$$
Thus, we have $n_{i}\geq n_0$.

We now prove  $n_{0}\geq n_{i}$.
All the elements in $GF(q)^{*}\cap{C}_{i}^{(N,r)}$ are denoted by
$$z_{1}=w_{1},$$
$$\vdots~~~~\vdots$$
$$z_{n_{i}}=w_{n_{i}},$$
where
$z_{1},\cdots,z_{n_{i}}\in GF(q)^{*}$
and $w_{1},\cdots,w_{n_{i}}\in {C}_{i}^{(N,r)}$.
Then
$z_jz_1^{-1}\in GF(q)^{*}$ and $w_jw_1^{-1}\in {C}_{0}^{(N,r)}$. Hence
$$z_jz_1^{-1}=w_jw_1^{-1}\in GF(q)^{*}\cap{C}_{0}^{(N,r)}\;(1\leq j\leq n_i).$$
Thus, we have $n_{0}\geq n_i$.

Then we have $n_{i}=n_0$, that is,
$\#(GF(q)^{*}\cap{C}_{i}^{(N,r)})
=
\frac{r-1}{[N,\frac{r-1}{q-1}]}$.
\end{proof}
\begin{lemma}\label{lem3.2}
Let $\beta\in GF(r)^{*}$, then

{\rm (a)}  $\beta GF(q)^{*}
\cap {C}_{i}^{(N,r)}\neq \emptyset$ if and only if $i\equiv i(\beta)\pmod {N_{2}}$.

{\rm (b)}  If $\beta GF(q)^{*}
\cap {C}_{i}^{(N,r)}\neq \emptyset$,
$\#(\beta GF(q)^{*}
\cap {C}_{i}^{(N,r)})
=\frac{r-1}{[N,\frac{r-1}{q-1}]}$.
\end{lemma}
\begin{proof}
{\rm (a)} $\beta GF(q)^{*}
\cap {C}_{i}^{(N,r)}\neq \emptyset$ if and only if
$$ GF(q)^{*}
\cap \beta^{-1}{C}_{i}^{(N,r)}
=GF(q)^{*}
\cap {C}_{i-i(\beta)}^{(N,r)}\neq \emptyset.$$
From Lemma \ref{lem3.1},
$GF(q)^{*}
\cap {C}_{i-i(\beta)}^{(N,r)}\neq \emptyset$ holds if and only if $$N_{2}|(i-i(\beta)).$$ that is $$i\equiv i(\beta)\pmod {N_{2}}.$$ Then $(a)$ is proved.

{\rm (b)} From $$\#(\beta GF(q)^{*}
\cap {C}_{i}^{(N,r)})
=\#(GF(q)^{*}
\cap {C}_{i-i(\beta)}^{(N,r)})$$ and Lemma \ref{lem3.1}, we can have $$\#(\beta GF(q)^{*}
\cap {C}_{i}^{(N,r)})
=\frac{r-1}{[N,\frac{r-1}{q-1}]}.$$
\end{proof}
\begin{lemma}\label{lem3.3}
Let $N'$ be a factor of $N$, then
$$
\eta_{i}^{(N,r)}+\eta_{i+N'}^{(N,r)}+\cdots
+\eta_{i+N'(\frac{N}{N'}-1)}^{(N,r)}=
\eta_{i}^{(N',r)}.
$$
\end{lemma}
\begin{proof}
This lemma can be proved from the definition of
$\eta_{i}^{(N,r)}$.
\end{proof}
\begin{theorem}\label{thm3.1}
 $\forall \beta\in GF(r)^{*}$, the Hamming weight of any codeword
$$
c(\beta)=(Tr_{r/q}(\beta), Tr_{r/q}(\beta\theta),
\ldots, Tr_{r/q}(\beta\theta^{n-1}))
$$
in the code $\mathcal{C}(r,N)$ is
$$
wt(c(\beta))=
\frac{(q-1)(r-1-N_{2}\eta_{i(\beta)}^{(N_{2},r)})}{qN}
=\frac{(q-1)(r-\eta_{i(\beta)}^{*(N_{2},r)})}
{qN}.
$$
\end{theorem}
\begin{proof}
The Hamming weight of $c(\beta)$ is
$$
wt(c(\beta))=n-\frac{Z(r,\beta)-1}{N},
$$
where
\begin{eqnarray*}
Z(r,\beta)&=&\frac{1}{q}\sum_{y\in GF(q)}\sum_{
x\in GF(r)}\zeta_p^{Tr_{q/p}(yTr_{r/q}(\beta x^N))}\\
&=&\frac{1}{q}[q+r-1+N\sum_{y\in GF(q)^*}
\sum_{x\in \mathcal{C}_{0}^{N,r}}\mu(y\beta x)].
\end{eqnarray*}
Note that
\begin{align*}
&\sum_{y\in GF(q)^*}
\sum_{x\in {C}_{0}^{(N,r)}}\mu(y\beta x)
=\sum_{y\in GF(q)^*}\mu(y\beta{C}_{0}^{N,r})\\
&=a_{0}(\beta)\eta_{0}^{(N,r)}+
a_{1}(\beta)\eta_{1}^{(N,r)}+\cdots+
a_{N-1}(\beta)\eta_{N-1}^{(N,r)},
\end{align*}
where $a_{i}(\beta)=\#(\beta GF(q)^{*}
\cap {C}_{i}^{(N,r)})$.

From Lemma
\ref{lem3.2} and Lemma \ref{lem3.3}, we have
\begin{align*}
&\sum_{y\in GF(q)^*}
\sum_{x\in {C}_{0}^{(N,r)}}\mu(y\beta x)\\
&=\frac{r-1}{[N,\frac{r-1}{q-1}]}
[\eta_{i(\beta)}^{(N,r)}+
\eta_{i(\beta)+N_{2}}^{(N,r)}+\cdots+
\eta_{i(\beta)+N_{2}(\frac{N}{N_{2}}-1)}^{(N,r)}]\\
&=\frac{r-1}{[N,\frac{r-1}{q-1}]}
\eta_{i(\beta)}^{(N_2,r)}.
\end{align*}
Then
$$
Z(r,\beta)=\frac{1}{q}[q+r-1+\frac{N(r-1)}{[N,\frac{r-1}{q-1}]}
\eta_{i(\beta)}^{(N_2,r)}].
$$
\begin{eqnarray*}
wt(c(\beta))&=& \frac{r-1-Z(r,\beta)+1}{N}\\
&=& \frac{r-\frac{1}{q}[q+r-1+\frac{N(r-1)}{[N,
\frac{r-1}{q-1}]}
\eta_{i(\beta)}^{(N_2,r)}]}{N}\\
&=& \frac{qr-q-r+1-\frac{N(r-1)}{[N,
\frac{r-1}{q-1}]}
\eta_{i(\beta)}^{(N_2,r)}}{qN}\\
&=& \frac{(q-1)(r-1)-(q-1)\frac{N\frac{r-1}{q-1}}{[N,
\frac{r-1}{q-1}]}
\eta_{i(\beta)}^{(N_2,r)}}{qN}\\
&=& \frac{q-1}{q}
\frac{r-1-N_{2}\eta_{i(\beta)}^{(N_{2},r)}}{N}\\
&=&\frac{q-1}{q}\frac{r-\eta_{i(\beta)}^{*(N_{2},r)}}{
N}.
\end{eqnarray*}
\end{proof}
\begin{theorem}\label{thm3.2}
Let the factorization of the reduced period polynomial $\psi_{(N_2,r)}^{*}$ be $$\psi_{(N_2,r)}^{*}=(X-\xi_{1})^{e_{1}}\cdots(X-
\xi_{t})
^{e_{t}},$$
where $e_{i}\geq 1$ and $e_{1}+
\cdots +e_{t}=N_{2}$, then $
\mathcal{C}(r,N)$ is a $[(r-1)/N,m]$ code and its weight distribution is
$$
1+\frac{e_{1}(r-1)}{N_{2}}x^{(q-1)(r-\xi_{1})/Nq}
+\cdots+\frac{e_{t}(r-1)}{N_{2}}x^
{(q-1)(r-\xi_{t})/Nq}.
$$
\end{theorem}
\begin{proof}
From Theorem \ref{thm3.1} , this theorem follows directly.
\end{proof}

Sometimes we also use the weight list
$$
wl=[0,(q-1)(r-\xi_{1})/Nq,\cdots,(q-1)(r-\xi_{t})/Nq ]
$$
and its corresponding weight distribution frequency
$$
Freq(0)=1, Freq(wl[i+1])=\frac{e_{i}(r-1)}{N_{2}}, i=1,\cdots,t.
$$
to describe the weight distribution.
\begin{corollary}
$\eta_{i}^{(N_{2},r)}$ and $\eta_{i}^{*(N_{2},r)}$
are integers. Further, $\psi_{(N_{2},r)}^{*}$ can be completely factorized over the rational field.
\end{corollary}
\begin{proof}
From the definition of $\eta_{i}^{(N_{2},r)}$ and $\eta_{i}^{*(N_{2},r)}$, they are algebraic integers. Then from Theorem \ref{thm3.1}, they are rational. Hence, they are integers and $\psi_{(N_{2},r)}^{*}$ can be completely factorized over the rational field.
\end{proof}

\begin{theorem}\label{thm3.3}
$\mathcal{C}(r,N)$ is  a code with only one nonzero weight if and only if $N_{2}=1$.
\end{theorem}
\begin{proof}
When $N_{2}=1$, $\mathcal{C}(r,N)$ is obviously  a code with only one nonzero weight.

We now prove that $\mathcal{C}(r,N)$ has at least two nonzero weights if $N_{2}>1$. Then we just need to prove $\eta_{i}^{(N_{2},r)}$ are not equal.

If $\eta_{1}^{(N_{2},r)}=\cdots =\eta_{N_{2}}^{(N_{2},r)}$, then
we have
$$
-1=\eta_{1}^{(N_{2},r)}+\cdots \eta_{N_{2}}^{(N_{2},r)}
=N_{2}\eta_{1}^{(N_{2},r)},
$$
that is, $\eta_{1}^{(N_{2},r)}=-\frac{1}{N_{2}}$. This contradicts that $\eta_{1}^{(N_{2},r)}$ is an integer. Then this theorem follows.
\end{proof}

From Theorem \ref{thm3.2}, $\mathcal{C}(r,N)$ have at most $N_{2}$ distinct nonzero weights. The following theorem gives a necessary and sufficient condition.
\begin{theorem}\label{thm3.4}
The number of distinct nonzero weights in $\mathcal{C}(r,N)$ achieves the maximum $N_{2}$
if and only if
$p\equiv 1\pmod {N_2}$.
\end{theorem}
\begin{proof}
Suppose $\mathcal{C}(r,N)$ has $N_{2}$ distinct nonzero
weights.  When $i\not\equiv j\pmod {N_2}$, then
$$\eta_{i}^{(N_{2},r)}\neq \eta_{j}^{(N_{2},r)}.$$ Note that $(C_{0}^{(N_{2},r)})^p=C_{0}^{
(N_{2},r)}$. Then $$C_{i}^{(N_{2},r)}=C_{pi}^{(N_{2},r)}.$$
Hence, $i\equiv pi\pmod {N_2}$, that is,
$$(p-1)i\equiv 0\pmod {N_2}.$$
For any $i$, $(p-1)i\equiv 0\pmod {N_2}$. Then we have $p\equiv 1\pmod {N_2}$.

We now prove that if $p\equiv 1\pmod{N_{2}}$,
$\mathcal{C}(r,N)$ has the maximal number $N_{2}$ of distinct nonzero weights. Then we require to  prove $c_{2}(\zeta_{r-1}^{0}\mod \;\wp)$,$c_{2}(\zeta_{r-1}^{1}\mod \;\wp)$,$
\cdots$,$c_{2}(\zeta_{r-1}^{N_{2}-1}\mod \;\wp)$
are all distinct.

From the definition of $\omega$, the character
$\chi$
satisfying $\chi^{N_{2}}=1$ can be represented by
$$
\chi=(\frac{1}{\omega})^{ni}, i=0,1,\cdots,
N_{2}-1,
$$
where $n=\frac{p-1}{N_{2}}p^0+
\frac{p-1}{N_{2}}p+\cdots+\frac{p-1}{N_{2}}p^{k-1}$. Then $s(ni)=\frac{k(p-1)i}{N_{2}}$.
Hence, from (\ref{for3})
$$
g(\omega^{-ni})\equiv \frac{(\zeta_{p}-1)^{\frac{k(p-1)}{N_{2}}i}}
{((\frac{p-1}{N_{2}}i)!)^{k}} \mod \;
\widetilde{\wp}^{\frac{k(p-1)}{N_{2}}i+1}.
$$
Then if $i=1$, $g(\omega^{-ni}) \mod
\widetilde{\wp}^{\frac{k(p-1)}{N_{2}}+1}
\equiv \frac{(\zeta_{p}-1)^{\frac{k(p-1)}{N_{2}}}}
{((\frac{p-1}{N_{2}})!)^{k}}$; if
$i=2,\cdots, N_{2}-1$, $g(\omega^{-ni}) \mod
\widetilde{\wp}^{\frac{k(p-1)}{N_{2}}+1}
\equiv 0$.

As a result,
\begin{eqnarray}\label{for4}
&&\sum_{\chi^{N_{2}}=1,\chi\neq 1}g(\chi)\widetilde{\chi}(\zeta_{r-1}^{j}\mod\;\wp)
=\sum_{i=1}^{N_{2}-1}g(\omega^{-ni})
\zeta_{r-1}^{nij}\nonumber\\
&&\equiv \frac{\zeta_{r-1}^{nj}}{((\frac{p-1}{N_{2}})!)^k}
(\zeta_{p}-1)^{\frac{k(p-1)}{N_{2}}}\pmod
{\widetilde{\wp}^{\frac{k(p-1)}{N_{2}}+1}}.
\end{eqnarray}
Since $\zeta_{r-1}^{nj}\pmod{\widetilde{\wp}}\;(j=1,2,\cdots,
N_{2}-1)$ are all different,
$\sum\limits_{\chi^{N_{2}}=1,\chi\neq 1}g(\chi)\widetilde{\chi}(\zeta_{r-1}^j\mod \; \wp)\;(j=1,2,\cdots,N_{2}-1)$  are all different.
Then the number of distinct nonzero $wt(c_{2}(\beta))$ is $N_{2}$.
This theorem is proved.
\end{proof}
From Theorem \ref{thm3.4}, we can get a divisible property of a nonzero codeword in $\mathcal{C}(r,N)$.
\begin{theorem}
Let $p\equiv 1\pmod{N_{2}}$ and
$\beta\in GF(r)^*$, then $q^{\frac{m}{N_{2}}-1}\|
wt(c(\beta))$.
\end{theorem}
\begin{proof}
Note that $p$ in $Z[\zeta_{p}]$ has the factorization
$$
(p)=(\zeta_{p}-1)^{p-1}.
$$
Then this theorem can be obtained obviously from
(\ref{for4})
in the proof of Theorem \ref{thm3.4}.
\end{proof}

\section{The Weight Distributions in  Cases $N_2=1,2,3,4$}
When $N_{2}=1$ or $N_{2}=2$,
Ding \cite{Ding2} determined the weight
distributions of $\mathcal{C}(r,N)$. His results are stated in
Theorem \ref{thm4.1} and Theorem \ref{thm4.5}.
\begin{theorem}\label{thm4.1}
Let $N_{2}=1$, then
$\mathcal{C}(r,N)$ is a
$[(r-1)/N,m, q^{m-1}(q-1)/N]$ code with the only nonzero weight $q^{m-1}(q-1)/N$.
\end{theorem}

\begin{theorem}\label{thm4.5}
Let $N_{2}=2$, then
$\mathcal{C}(r,N)$ is a
$[(r-1)/N,m, (q-1)(r-r^{1/2})/Nq]$ two weight code with the weight distribution
$$
1+\frac{r-1}{2}x^{(q-1)(r-r^{1/2})/Nq}+
\frac{r-1}{2}x^{(q-1)(r+r^{1/2})/Nq}.
$$
\end{theorem}

\begin{theorem}\label{thm4.6}
Let $N_{2}=3$ and $p\equiv 1 \pmod 3$, then
$k\equiv 0\pmod 3$ and
$\mathcal{C}(r,N)$ is a
$[(r-1)/N,m]$ code with the weight distribution
\begin{align*}
1+\frac{r-1}{3}x^{(q-1)(r-cr^{1/3})/Nq}+
\frac{r-1}{3}x^{(q-1)(r+\frac{1}{2}(c+9d)
r^{1/3})/Nq}\\
+
\frac{r-1}{3}x^{(q-1)(r+\frac{1}{2}(c-9d)
r^{1/3})/Nq},
\end{align*}
where $c$ and $d$ satisfy
$4p^{k/3}=c^2+27d^2$, $c\equiv 1\pmod 3$ and
$gcd(c,p)=1$.
\end{theorem}
\begin{proof}
From Theorem \ref{thm3.2} and (a) in Lemma
\ref{lem2}, this theorem follows immediately.
\end{proof}
\begin{theorem}\label{thm4.7}
Let $N_{2}=3$ and $p\equiv 2\pmod 3$, then
$k\equiv 0\pmod 2$ and
$\mathcal{C}(r,N)$ is a
$[(r-1)/N,m]$ two weight code with the weight distribution
\begin{align*}
1&+\frac{r-1}{3}x^{(q-1)(r+(-1)^{k/2}2r^{1/2})/Nq}\\
&+
\frac{2(r-1)}{3}x^{(q-1)(r-(-1)^{k/2}2r^{1/2})/Nq}.
\end{align*}
\end{theorem}
\begin{proof}
From Theorem \ref{thm3.2} and (b) in Lemma
\ref{lem2}, this theorem follows immediately.
\end{proof}
\begin{theorem}\label{thm4.12}
Let $N_{2}=4$ and $p\equiv 1 \pmod 4$, then
$k\equiv 0\pmod 4$ and
$\mathcal{C}(r,N)$ is a
$[(r-1)/N,m]$ code with the weight distribution
\begin{eqnarray*}
1&+&\frac{r-1}{4}x^{(q-1)(r+\sqrt{r}+2r^{1/4}u)/Nq}
\\
&+&
\frac{r-1}{4}x^{(q-1)(r+\sqrt{r}-2r^{1/4}u)/Nq}\\
&+&
\frac{r-1}{4}x^{(q-1)(r-\sqrt{r}+4r^{1/4}v)/Nq}
\\
&+&
\frac{r-1}{4}x^{(q-1)(r-\sqrt{r}-4r^{1/4}v)/Nq},
\end{eqnarray*}
where $u$ and $v$ satisfy
$p^{k/2}=u^2+4v^2$, $u\equiv 1\pmod 4$ and
$gcd(u,p)=1$.
\end{theorem}
\begin{proof}
From Theorem \ref{thm3.2} and (a) in Lemma
\ref{lem3}, this theorem follows immediately.
\end{proof}
\begin{theorem}\label{thm4.13}
Let $N_{2}=4$ and $p\equiv 3\pmod 4$, then
$k\equiv 0\pmod 2$ and
$\mathcal{C}(r,N)$ is a
$[(r-1)/N,m]$ two weight code with the weight distribution
\begin{align*}
1&+\frac{r-1}{4}x^{(q-1)(r+(-1)^{k/2}3r^{1/2})/Nq}\\
&+
\frac{3(r-1)}{4}x^{(q-1)(r-(-1)^{k/2}r^{1/2})/Nq}.
\end{align*}
\end{theorem}
\begin{proof}
From Theorem \ref{thm3.2} and (b) in Lemma
\ref{lem3}, this theorem follows immediately.
\end{proof}

\section{The Weight Distributions in  Cases N=5,6,8,12}
\begin{theorem}\label{thm5.1}
Let $N=5$ and $N_{2}=1$, then
$\mathcal{C}(r,5)$ is a
$[(r-1)/5,m, q^{m-1}(q-1)/5]$ code with the only nonzero weight $q^{m-1}(q-1)/5$.
\end{theorem}
\begin{proof}
From the weight distribution when
$N_{2}=1$ in Ding \cite{Ding2},
This theorem follows immediately.
\end{proof}
\begin{example}
Let $~q=11$ and let $~m=2$. Then the set
$~\mathcal{C}(r,5)$
 is a $[24
 ,2,22]$ code over $~GF(11)$ with the weight distribution
\begin{align*}
 1+120x^{22}
\end{align*}
\end{example}

\begin{theorem}\label{thm5.2}
Let $N=5$, $N_{2}=5$ and $p\equiv 1\pmod 5$, then
$\mathcal{C}$ is a $[(r-1)/5,m]$ code with weight list
\begin{align*}
&wl=\\
&[0,\frac{q-1}{qN}(r+\frac{\sqrt[5]{r}}{16}(x^3-25L)
), \frac{q-1}{qN}(r-\frac{\sqrt[5]{r}}{64}(x^3-25M)
),\\
&\frac{q-1}{qN}(r-\frac{\sqrt[5]{r}}{64}\sigma(x^3-25M)
),
\frac{q-1}{qN}(r-\frac{\sqrt[5]{r}}{64}\sigma^2
(x^3-25M)
),\\
&\frac{q-1}{qN}(r-\frac{\sqrt[5]{r}}{64}\sigma^3
(x^3-25M)
)]
\end{align*}
and the weight distribution frequency
$$
Freq(0)=1, Freq(wl[i])=\frac{r-1}{5}~(i=2,\cdots,6).
$$
\end{theorem}
\begin{proof}
From Theorem \ref{thm3.2} and (a) in Lemma
\ref{lem4}, this theorem follows immediately.
\end{proof}
\begin{example}
Let $~q=11$ and let $~m=5$. Then the set
$~\mathcal{C}(r,5)$
 is a $[32210
 ,5,29050]$ code over $~GF(11)$ with the weight distribution
\begin{align*}
 1&+32210x^{29050}+32210x^{29200}+32210x^{29300}\\
 &+32210x^{29400}+32210x^{29460}
\end{align*}
\end{example}

\begin{theorem}\label{thm5.3}
Let $N=5$, $N_{2}=5$ and
$p\not \equiv 1\pmod 5$. Let $k=sm$.

{\rm (a)} If $k/ord_5p$ is
odd, $\mathcal{C}(r,5)$ is a
$[(r-1)/5,m,(q-1)(r-4r^{1/2})/5q]$ code with the weight distribution
$$
1+\frac{r-1}{5}x^{(q-1)(r-4r^{1/2})/5q}+
\frac{4(r-1)}{5}x^{(q-1)(r+r^{1/2})/5q}.
$$

{\rm (b)} If $k/ord_5p$ is
even,  $\mathcal{C}(r,5)$ is a
$[(r-1)/5,m,(q-1)(r-r^{1/2})/5q]$ code with the weight distribution
$$
1+\frac{4(r-1)}{5}x^{(q-1)(r-r^{1/2})/5q}+
\frac{r-1}{5}x^{(q-1)(r+4r^{1/2})/5q}.
$$
\end{theorem}
\begin{proof}
This theorem is
the semi-primitive case in \cite{BM1}.
\end{proof}
\begin{example}
Let $~q=7^2$ and let $~m=2$. Then the set
$~\mathcal{C}(r,5)$
 is a $[480
 ,2,432]$ code over $~GF(7^2)$ with the weight distribution
\begin{align*}
 1+480x^{432}+1920x^{480}
\end{align*}

Let $q=19$ and let $m=4$. Then the set
$~\mathcal{C}(r,5)$
 is a $[26064
 ,4,24624]$ code over $~GF(19)$ with the weight distribution
\begin{align*}
 1+104256x^{24624}+26064x^{24966}
\end{align*}
\end{example}

\begin{theorem}\label{thm5.4}
Let $N=6$ and $N_{2}=1$, then
$\mathcal{C}(r,6)$ is a
$[(r-1)/6,m, q^{m-1}(q-1)/6]$ code with the only nonzero weight $q^{m-1}(q-1)/6$.
\end{theorem}
\begin{proof}
This theorem follows immediately from
Ding \cite{Ding2} for the case $N_{2}=1$.
\end{proof}
\begin{example}
Let $~q=7$ and let $~m=5$. Then the set
$~\mathcal{C}(r,6)$
 is a $[2801
 ,5,2401]$ code over $~GF(7)$ with the weight distribution
\begin{align*}
 1+16806x^{2401}
\end{align*}
\end{example}

\begin{theorem}\label{thm5.5}
Let $N=6$ and $N_{2}=2$, then
$\mathcal{C}(r,6)$ is a
$[(r-1)/6,m, (q-1)(r-r^{1/2})/6q]$ two weight code with the weight distribution
$$
1+\frac{r-1}{2}x^{(q-1)(r-r^{1/2})/6q}+
\frac{r-1}{2}x^{(q-1)(r+r^{1/2})/6q}.
$$
\end{theorem}
\begin{proof}
This theorem follows immediately from
Ding \cite{Ding2} for the case $N_{2}=2$.
\end{proof}
\begin{example}
Let $~q=7$ and let $~m=2$. Then the set
$~\mathcal{C}(r,6)$
 is a $[8
 ,2,6]$ code over $~GF(7)$ with the weight distribution
\begin{align*}
 1+24x^{6}+24x^{8}
\end{align*}
\end{example}

\begin{theorem}\label{thm5.6}
Let $N=6$, $N_{2}=3$ and $p\equiv 1 \pmod 3$, then
$k\equiv 0\pmod 3$ and
$\mathcal{C}(r,6)$ is a
$[(r-1)/6,m]$ code with the weight distribution
\begin{align*}
1+\frac{r-1}{3}x^{(q-1)(r-cr^{1/3})/6q}+
\frac{r-1}{3}x^{(q-1)(r+\frac{1}{2}(c+9d)
r^{1/3})/6q}
\\+
\frac{r-1}{3}x^{(q-1)(r+\frac{1}{2}(c-9d)
r^{1/3})/6q},
\end{align*}
where $c$ and $d$ satisfy
$4p^{k/3}=c^2+27d^2$, $c\equiv 1\pmod 3$ and
$gcd(c,p)=1$.
\end{theorem}
\begin{proof}
This theorem follows immediately
from the proof in  Theorem \ref{thm4.6}.
\end{proof}
\begin{example}
Let $~q=7$ and let $~m=3$. Then the set
$~\mathcal{C}(r,6)$
 is a $[57
 ,3,45]$ code over $~GF(7)$ with the weight distribution
\begin{align*}
 1+114x^{45}+114x^{48}+114x^{54}
\end{align*}
\end{example}

\begin{theorem}\label{thm5.7}
Let $N=6$, $N_{2}=3$ and $p\equiv 2\pmod 3$, then
$k\equiv 0\pmod 2$ and
$\mathcal{C}(r,6)$ is a
$[(r-1)/6,m]$ two weight code with the weight distribution
\begin{align*}
1&+\frac{r-1}{3}x^{(q-1)(r+(-1)^{k/2}2r^{1/2})/6q}
\\
&+
\frac{2(r-1)}{3}x^{(q-1)(r-(-1)^{k/2}2r^{1/2})/6q}.
\end{align*}
\end{theorem}
\begin{proof}
This theorem follows immediately
from the proof in  Theorem \ref{thm4.7}.
\end{proof}
\begin{example}
Let $~q=5^2$ and let $~m=3$. Then the set
$~\mathcal{C}(r,6)$
 is a $[2604
 ,3,2460]$ code over $~GF(5^2)$ with the weight distribution
\begin{align*}
 1+5208x^{2460}+10416x^{2520}
\end{align*}
\end{example}

\begin{theorem}\label{thm5.8}
Let $N=6$, $N_{2}=6$ and $p\equiv 1 \pmod 6$, then
$k\equiv 0\pmod 6$ and
$\mathcal{C}(r,6)$ is a
$[(r-1)/6,m]$ code with the weight distribution
$$
1+\frac{r-1}{6}x^{(q-1)(r-\eta_{1}^{*(6,r)})/6q}+
\cdots+
\frac{r-1}{6}x^{(q-1)(r-\eta_{6}^{*(6,r)})/6q}.
$$
\end{theorem}
\begin{proof}
From Theorem \ref{thm3.2} and
(a) in Lemma \ref{lem5}, this theorem
follows immediately.
\end{proof}
\begin{example}
Let $~q=7$ and let $~m=6$. Then the set
$~\mathcal{C}(r,6)$
 is a $[19608
 ,6,16596]$ code over $~GF(7)$ with the weight distribution
\begin{align*}
 1&+19608x^{16596}+19608x^{16776}+19608x^{16812}\\
 &+19608x^{16836}+19608x^{16866}+19608x^{16956}
\end{align*}
\end{example}

\begin{theorem}\label{thm5.9}
Let $N=6$, $N_{2}=6$ and $p\equiv 5\pmod 6$, then
$\mathcal{C}(r,6)$ is a
$[(r-1)/6,m]$ two weight code with the weight distribution
\begin{align*}
1&+\frac{5(r-1)}{6}x^{(q-1)(r-(-1)^{k/2}p^{k/2})/6q}\\
&+
\frac{r-1}{6}x^{(q-1)(r+5(-1)^{k/2}p^{k/2})/6q}.
\end{align*}
\end{theorem}
\begin{proof}
This theorem is
the semi-primitive case in \cite{BM1}.
\end{proof}
\begin{example}
Let $~q=11$ and let $~m=2$. Then the set
$~\mathcal{C}(r,6)$
 is a $[20
 ,2,10]$ code over $~GF(11)$ with the weight distribution
\begin{align*}
 1+20x^{10}+100x^{20}
\end{align*}
\end{example}

\begin{theorem}\label{thm5.10}
Let $N=8$ and $N_{2}=1$, then
$\mathcal{C}(r,8)$ is a
$[(r-1)/8,m, q^{m-1}(q-1)/8]$ code with the only nonzero weight $q^{m-1}(q-1)/8$.
\end{theorem}
\begin{proof}
This theorem follows immediately from
Ding \cite{Ding2} for the case $N_{2}=1$.
\end{proof}
\begin{example}
Let $~q=17$ and let $~m=3$. Then the set
$~\mathcal{C}(r,8)$
 is a $[614
 ,3,578]$ code over $~GF(17)$ with the weight distribution
\begin{align*}
 1+4912x^{578}
\end{align*}
\end{example}

\begin{theorem}\label{thm5.11}
Let $N=8$ and $N_{2}=2$, then
$\mathcal{C}(r,8)$ is a
$[(r-1)/8,m, (q-1)(r-r^{1/2})/8q]$ two weight code with the weight distribution
$$
1+\frac{r-1}{2}x^{(q-1)(r-r^{1/2})/8q}+
\frac{r-1}{2}x^{(q-1)(r+r^{1/2})/8q}.
$$
\end{theorem}
\begin{proof}
This theorem follows immediately from
Ding \cite{Ding2} for the case $N_{2}=2$.
\end{proof}
\begin{example}
Let $~q=7^2$ and let $~m=2$. Then the set
$~\mathcal{C}(r,8)$
 is a $[300
 ,2,288]$ code over $~GF(7^2)$ with the weight distribution
\begin{align*}
 1+1200x^{288}+1200x^{300}
\end{align*}
\end{example}

\begin{theorem}\label{thm5.12}
Let $N=8$, $N_{2}=4$ and $p\equiv 1 \pmod 4$, then
$k\equiv 0\pmod 4$ and
$\mathcal{C}(r,8)$ is a
$[(r-1)/8,m]$ code with the weight distribution
\begin{align*}
1&+\frac{r-1}{4}x^{(q-1)(r+\sqrt{r}+2r^{1/4}u)/8q}\\
&+
\frac{r-1}{4}x^{(q-1)(r+\sqrt{r}-2r^{1/4}u)/8q}\\
&+\frac{r-1}{4}x^{(q-1)(r-\sqrt{r}+4r^{1/4}v)/8q}\\
&+\frac{r-1}{4}x^{(q-1)(r-\sqrt{r}-4r^{1/4}v)/8q},
\end{align*}
where $u$ and $v$ satisfy
$p^{k/2}=u^2+4v^2$, $u\equiv 1\pmod 4$ and
$gcd(u,p)=1$.
\end{theorem}
\begin{proof}
This theorem follows immediately from
the proof in Theorem \ref{thm4.12}.
\end{proof}
\begin{example}
Let $~q=17$ and let $~m=4$. Then the set
$~\mathcal{C}(r,8)$
 is a $[10440
 ,4,9760]$ code over $~GF(17)$ with the weight distribution
\begin{align*}
 1+20880x^{9760}+20880x^{9800}+20880x^{9824}+20880x^{9920}
\end{align*}
\end{example}

\begin{theorem}\label{thm5.13}
Let $N=8$, $N_{2}=4$ and $p\equiv 3\pmod 4$, then
$k\equiv 0\pmod 2$ and
$\mathcal{C}(r,8)$ is a
$[(r-1)/8,m]$ two weight code with the weight distribution
\begin{align*}
1&+\frac{r-1}{4}x^{(q-1)(r+(-1)^{k/2}3r^{1/2})/8q}
\\&+
\frac{3(r-1)}{4}x^{(q-1)(r-(-1)^{k/2}r^{1/2})/8q}.
\end{align*}
\end{theorem}
\begin{proof}
This theorem follows immediately from
the proof in Theorem \ref{thm4.13}.
\end{proof}
\begin{example}
Let $~q=11$ and let $~m=2$. Then the set
$~\mathcal{C}(r,8)$
 is a $[15
 ,2,10]$ code over $~GF(11)$ with the weight distribution
\begin{align*}
 1+30x^{10}+90x^{15}
\end{align*}
\end{example}

\begin{theorem}\label{thm5.13.1}
Let $N=8$, $N_{2}=8$ and $p\equiv 1\pmod 8$, then
$k\equiv 0\pmod 8$ and
$\mathcal{C}(r,8)$ is a
$[(r-1)/8,m]$  code with the weight distribution
$$
1+\frac{r-1}{8}x^{(q-1)(r-\eta_{1}^{*(8,r)})/8q}+
\cdots+
\frac{r-1}{8}x^{(q-1)(r-\eta_{8}^{*(8,r)})/8q},
$$
where $\eta_{j}^{*(8,r)}=-p^{k/2}-
p^{k/4}Q_{j,k/2}-p^{k/8}AB$ and
$A,B$ are defined as follows: If
$j$ is even, $A=Q_{j/2,k/4}$ and $B=T_{k/2}$;
if $j$ is odd, $A=(-1)^{[j/4]}Q_{0,k/4}
+(-1)^{[(j-2)/4]}P_{0,k/4}$ and $B=S_{k/2}$.
\end{theorem}
\begin{proof}
From Theorem \ref{thm3.2} and
(a) in Lemma \ref{lem6}, this theorem follows immediately.
\end{proof}
\begin{example}
Let $~q=17$ and let $~m=8$. Then the set
$~\mathcal{C}(r,8)$
 is a $[871969680,8,820657856]$ code over $~GF(17)$ with the weight list
\begin{align*}
 wl:=[&0,820657856, 820663680, 820666436, 820675268,\\
  &820694592, 820702148, 820704836, 820732560]
\end{align*}
 and the weight distribution frequency
 \begin{align*}
 Freq(0)&=1\\
 Freq(wl[i])&=871969680,~i=2,\cdots 9.
 \end{align*}
\end{example}

\begin{theorem}\label{thm5.13.2}
Let $N=8$, $N_{2}=8$ and $p\equiv 3\pmod 8$, then
$k\equiv 0\pmod 4$ and
$\mathcal{C}(r,8)$ is a
$[(r-1)/8,m]$  code with the weight distribution
\begin{align*}
1&+\frac{r-1}{4}x^{(q-1)(r-\xi_{1})/8q}
+\frac{r-1}{4}x^{(q-1)(r-\xi_{2})/8q}\\&
+\frac{r-1}{4}x^{(q-1)(r-\xi_{3})/8q}
+\frac{r-1}{8}x^{(q-1)(r-\xi_{4})/8q}\\
&+\frac{r-1}{8}x^{(q-1)(r-\xi_{5})/8q},
\end{align*}
where
$\xi_{1}=-2p^{k/2}S_{k/2}+p^{k/2}$, $\xi_{2}=2p^{k/2}S_{k/2}+p^{k/2}$,
$\xi_{3}=p^{k/2}$, $\xi_{4}=2p^{k/4}T_{k/2}-3p^{k/2}$ and
$\xi_{5}=-2p^{k/4}T_{k/2}-3p^{k/2}.$
\end{theorem}
\begin{proof}
From Theorem \ref{thm3.2} and
(b) in Lemma \ref{lem6}, this theorem follows immediately.
\end{proof}
\begin{example}
Let $~q=3$ and let $~m=4$. Then the set
$~\mathcal{C}(r,8)$
 is a $[10
 ,4,4]$ code over $~GF(3)$ with the weight distribution
\begin{align*}
 1+20x^{4}+20x^{6}+30x^{8}+10x^{10}
\end{align*}
\end{example}

\begin{theorem}\label{thm5.13.3}
Let $N=8$, $N_{2}=8$ and $p\equiv 5\pmod 8$, then
$k\equiv 0\pmod 8$ and
$\mathcal{C}(r,8)$ is a
$[(r-1)/8,m]$  code with the weight distribution
\begin{align*}
1&+\frac{r-1}{4}x^{(q-1)(r-\xi_{1})/8q}
+\frac{r-1}{4}x^{(q-1)(r-\xi_{2})/8q}\\
&+\frac{r-1}{8}x^{(q-1)(r-\xi_{3})/8q}
+\frac{r-1}{8}x^{(q-1)(r-\xi_{4})/8q}\\
&+\frac{r-1}{8}x^{(q-1)(r-\xi_{5})/8q}
+\frac{r-1}{8}x^{(q-1)(r-\xi_{6})/8q}
,
\end{align*}
where
$\xi_{1}=p^{k/2}-p^{k/4}P_{k/2}$, $\xi_{2}=p^{k/2}+p^{k/4}P_{k/2}$,
$\xi_{3}=-p^{k/2}-p^{k/4}Q_{k/2}-2p^{3k/8}Q_{k/4}$,  $\xi_{4}=-p^{k/2}-p^{k/4}Q_{k/2}+2p^{3k/8}Q_{k/4}$,  $\xi_{5}=-p^{k/2}+p^{k/4}Q_{k/2}-2p^{3k/8}P_{k/4}$ and  $\xi_{6}=-p^{k/2}+p^{k/4}Q_{k/2}+2p^{3k/8}P_{k/4}$.
\end{theorem}
\begin{proof}
From Theorem \ref{thm3.2} and
(c) in Lemma \ref{lem6}, this theorem follows immediately.
\end{proof}
\begin{example}
Let $~q=5$ and let $~m=8$. Then the set
$~\mathcal{C}(r,8)$
 is a $[48828
 ,8,38880]$ code over $~GF(5)$ with the weight distribution
\begin{align*}
 1&+97656x^{38880}+48828x^{38940}+48828x^{38960}\\
 &+97656x^{39120}+48828x^{39240}+48828x^{39360}
\end{align*}
\end{example}

\begin{theorem}\label{thm5.13.4}
Let $N=8$, $N_{2}=8$ and $p\equiv 7\pmod 8$, then
$k\equiv 0\pmod 2$ and
$\mathcal{C}(r,8)$ is a
$[(r-1)/8,m]$  two weight code with the weight distribution
\begin{align*}
1&+\frac{7(r-1)}{8}x^{(q-1)(r-(-1)^{k/2}p^{k/2})/8q}\\
&+\frac{r-1}{8}x^{(q-1)(r+7(-1)^{k/2}p^{k/2})/8q}.
\end{align*}
\end{theorem}
\begin{proof}
This theorem is
the semi-primitive case in \cite{BM1}.
\end{proof}
\begin{example}
Let $~q=7$ and let $~m=4$. Then the set
$~\mathcal{C}(r,8)$
 is a $[300
 ,4,252]$ code over $~GF(7)$ with the weight distribution
\begin{align*}
 1+2100x^{252}+300x^{294}
\end{align*}
\end{example}

\begin{theorem}\label{thm5.14}
Let $N=12$ and $N_{2}=1$, then
$\mathcal{C}(r,12)$ is a
$[(r-1)/12,m, q^{m-1}(q-1)/12]$ code with the only nonzero weight $q^{m-1}(q-1)/12$.
\end{theorem}
\begin{proof}
This theorem follows immediately from
Ding \cite{Ding2} for the case $N_{2}=1$.
\end{proof}
\begin{example}
Let $~q=13$ and let $~m=5$. Then the set
$~\mathcal{C}(r,12)$
 is a $[30941
 ,5,28561]$ code over $~GF(13)$ with the weight distribution
\begin{align*}
 1+371292x^{28561}
\end{align*}
\end{example}

\begin{theorem}\label{thm5.15}
Let $N=12$ and $N_{2}=2$, then
$\mathcal{C}(r,12)$ is a
$[(r-1)/12,m, (q-1)(r-r^{1/2})/12q]$ two weight code with the weight distribution
$$
1+\frac{r-1}{2}x^{(q-1)(r-r^{1/2})/12q}+
\frac{r-1}{2}x^{(q-1)(r+r^{1/2})/12q}.
$$
\end{theorem}
\begin{proof}
This theorem follows immediately from
Ding \cite{Ding2} for the case $N_{2}=2$.
\end{proof}
\begin{example}
Let $~q=5^2$ and let $~m=2$. Then the set
$~\mathcal{C}(r,12)$
 is a $[52
 ,2,48]$ code over $~GF(5^2)$ with the weight distribution
\begin{align*}
 1+312x^{48}+312x^{52}
\end{align*}
\end{example}

\begin{theorem}\label{thm5.16}
Let $N=12$, $N_{2}=3$ and $p\equiv 1 \pmod 3$, then
$k\equiv 0\pmod 3$ and
$\mathcal{C}(r,12)$ is a
$[(r-1)/12,m]$ code with the weight distribution
\begin{align*}
1&+\frac{r-1}{3}x^{(q-1)(r-cr^{1/3})/12q}+
\frac{r-1}{3}x^{(q-1)(r+\frac{1}{2}(c+9d)r^{1/3})
/12q}
\\&+
\frac{r-1}{3}x^{(q-1)(r+\frac{1}{2}(c-9d)r^{1/3})/12q},
\end{align*}
where $c$ and $d$ satisfy
$4p^{k/3}=c^2+27d^2$, $u\equiv 1\pmod 3$ and
$gcd(c,p)=1$.
\end{theorem}
\begin{proof}
This theorem follows immediately from
Theorem \ref{thm4.6}.
\end{proof}
\begin{example}
Let $~q=13$ and let $~m=3$. Then the set
$~\mathcal{C}(r,12)$
 is a $[183
 ,3,162]$ code over $~GF(13)$ with the weight distribution
\begin{align*}
 1+732x^{162}+732x^{171}+732x^{174}
\end{align*}
\end{example}

\begin{theorem}\label{thm5.17}
Let $N=12$, $N_{2}=3$ and $p\equiv 2\pmod 3$, then
$k\equiv 0\pmod 2$ and
$\mathcal{C}(r,12)$ is a
$[(r-1)/12,m]$ two weight code with the weight distribution
\begin{align*}
1&+\frac{r-1}{3}x^{(q-1)(r+(-1)^{k/2}2r^{1/2})/12q}
\\&+
\frac{2(r-1)}{3}x^{(q-1)(r-(-1)^{k/2}r^{1/2})/12q}.
\end{align*}
\end{theorem}
\begin{proof}
This theorem follows immediately from
Theorem \ref{thm4.7}.
\end{proof}
\begin{example}
Let $~q=5^2$ and let $~m=3$. Then the set
$~\mathcal{C}(r,12)$
 is a $[1302
 ,3,1230]$ code over $~GF(5^2)$ with the weight distribution
\begin{align*}
 1+5202x^{1230}+10416x^{1260}
\end{align*}
\end{example}

\begin{theorem}\label{thm5.18}
Let $N=12$, $N_{2}=4$ and $p\equiv 1 \pmod 4$, then
$k\equiv 0\pmod 4$ and
$\mathcal{C}(r,12)$ is a
$[(r-1)/12,m]$ code with the weight distribution
\begin{align*}
1&+\frac{r-1}{4}x^{(q-1)(r+\sqrt{r}+2r^{1/4}u)/12q}\\&
+
\frac{r-1}{4}x^{(q-1)(r+\sqrt{r}-2r^{1/4}u)/12q}\\
&+
\frac{r-1}{4}x^{(q-1)(r-\sqrt{r}+4r^{1/4}v)/12q}\\&+
\frac{r-1}{4}x^{(q-1)(r-\sqrt{r}-4r^{1/4}v)/12q},
\end{align*}
where $u$ and $v$ satisfy
$p^{k/2}=u^2+4v^2$, $u\equiv 1\pmod 4$ and
$gcd(u,p)=1$.
\end{theorem}
\begin{proof}
This theorem follows immediately from
Theorem \ref{thm4.12}.
\end{proof}
\begin{example}
Let $q=13$ and let $m=4$. Then the set
$~\mathcal{C}(r,12)$
 is a $[2380
 ,4,2160]$ code over $~GF(13)$ with the weight distribution
\begin{align*}
 1+7140x^{2160}+7140x^{2200}+7140x^{2208}+7140x^{2220}
\end{align*}
\end{example}

\begin{theorem}\label{thm5.19}
Let $N=12$, $N_{2}=4$ and $p\equiv 3\pmod 4$, then
$k\equiv 0\pmod 2$ and
$\mathcal{C}(r,12)$ is a
$[(r-1)/12,m]$ two weight code with the weight distribution
\begin{align*}
1&+\frac{r-1}{4}x^{(q-1)(r+(-1)^{k/2}3r^{1/2})/12q}\\
&+
\frac{3(r-1)}{4}x^{(q-1)(r-(-1)^{k/2}r^{1/2})/12q}.
\end{align*}
\end{theorem}
\begin{proof}
This theorem follows immediately from
Theorem \ref{thm4.13}.
\end{proof}
\begin{example}
Let $~q=7$ and let $~m=2$. Then the set
$~\mathcal{C}(r,12)$
 is a $[4
 ,2,2]$ code over $~GF(7)$ with the weight distribution
\begin{align*}
 1+12x^{2}+36x^{4}
\end{align*}
\end{example}

\begin{theorem}\label{thm5.20}
Let $N=12$, $N_{2}=6$ and $p\equiv 1\pmod 6$, then
$k\equiv 0\pmod 6$ and
$\mathcal{C}(r,12)$ is a
$[(r-1)/12,m]$ code with the weight distribution
$$
1+\frac{r-1}{6}x^{(q-1)(r-\eta_{1}^{*(6,r)})/12q}+
\cdots+
\frac{r-1}{6}x^{(q-1)(r-\eta_{6}^{*(6,r)})/12q}.
$$
\end{theorem}
\begin{proof}
From Theorem \ref{thm3.2} and
(a) in Lemma \ref{lem5}, this theorem follows
immediately.
\end{proof}
\begin{example}
Let $q=13$ and let $m=6$. Then the set
$~\mathcal{C}(r,12)$
 is a $[402234,6,370692]$ code over $~GF(13)$ with the weight list
\begin{align*}
 wl:=[0,370692, 371112, 371232, 371322, 371448, 371952]
\end{align*}
 and the weight distribution frequency
 \begin{align*}
 Freq(0)&=1\\
 Freq(wl[i])&=804468,~i=2,\cdots 7.
 \end{align*}
\end{example}

\begin{theorem}\label{thm5.21}
Let $N=12$, $N_{2}=6$ and $p\equiv 5\pmod 6$, then
$\mathcal{C}(r,12)$ is a
$[(r-1)/12,m]$ two weight code with the weight distribution
\begin{align*}
1&+\frac{5(r-1)}{6}x^{(q-1)(r-(-1)^{k/2}p^{k/2})/12q}\\
&+
\frac{r-1}{6}x^{(q-1)(r+5(-1)^{k/2}p^{k/2})/12q}.
\end{align*}
\end{theorem}
\begin{proof}
This theorem is
the semi-primitive case in \cite{BM1}.
\end{proof}
\begin{example}
Let $~q=17$ and let $~m=2$. Then the set
$~\mathcal{C}(r,12)$
 is a $[24
 ,2,16]$ code over $~GF(17)$ with the weight distribution
\begin{align*}
 1+48x^{16}+240x^{24}
\end{align*}
\end{example}

\begin{theorem}\label{thm5.21.1}
Let $N=12$, $N_{2}=12$ and $p\equiv 1\pmod {12}$, then $k\equiv 0\pmod {12}$ and
$\mathcal{C}(r,12)$ is a
$[(r-1)/12,m]$  code with the weight distribution
$$
1+\frac{r-1}{12}x^{(q-1)(r-\eta_{1}^{*(12,r)})/12q}+
\cdots
\frac{r-1}{12}x^{(q-1)(r-\eta_{12}^{*(12,r)})/12q},
$$
where $\eta_{j}^{*(12,r)}=-p^{k/12}Q_{j,k/2}V_{-j,k/3}
-p^{k/4}Q_{j,k/2}
-p^{k/6}V_{j,2k/3}-p^{k/3}V_{2j,k/3}
-(-1)^jp^{k/2}$.
\end{theorem}
\begin{proof}
From Theorem \ref{thm3.2} and
(a) in Lemma \ref{lem7}, this theorem follows
immediately.
\end{proof}
\begin{example}
Let $~q=13$ and let $~m=12$. Then the set
$~\mathcal{C}(r,12)$
 is a $[1941507093540,12,1792157710608]$ code over $~GF(13)$ with the weight list
\begin{align*}
 wl:=&[0,1792157710608,1792159338564,1792159386480,\\
 &1792159451424, 1792160074992,1792160674272,\\
 &1792160747136, 1792160770896,1792160847072, \\ &1792161442512, 1792161902664, 1792162381824]
\end{align*}
 and the weight distribution frequency
 \begin{align*}
 Freq(0)&=1\\
 Freq(wl[i])&=1941507093540,~i=2,\cdots 13.
 \end{align*}
\end{example}

\begin{theorem}\label{thm5.21.2}
Let $N=12$, $N_{2}=12$ and $p\equiv 5\pmod {12}$, then
$k\equiv 0\pmod 4$ and
$\mathcal{C}(r,12)$ is a
$[(r-1)/12,m]$  code with the weight distribution
\begin{align*}
1&+\frac{r-1}{6}x^{(q-1)(r-\xi_{1})/12q}
+\frac{r-1}{6}x^{(q-1)(r-\xi_{2})/12q}\\
&+\frac{r-1}{6}x^{(q-1)(r-\xi_{3})/12q}+
\frac{r-1}{6}x^{(q-1)(r-\xi_{4})/12q}\\
&+\frac{r-1}{12}x^{(q-1)(r-\xi_{5})/12q}+\cdots
+\frac{r-1}{12}x^{(q-1)(r-\xi_{8})/12q}
,
\end{align*}
where
$\xi_{1}=Q_{k/2}p^{k/4}((-1)^{k/4}-1)+p^{k/2}$,  $\xi_{2}=-Q_{k/2}p^{k/4}((-1)^{k/4}-1)+p^{k/2}$,
$\xi_{3}=P_{k/2}p^{k/4}((-1)^{k/4}+1)+p^{k/2}$, $\xi_{4}=-P_{k/2}p^{k/4}((-1)^{k/4}+1)+p^{k/2}$,
 $\xi_{5}=P_{k/2}p^{k/4}(2(-1)^{k/4}-1)+p^{k/2}$,  $\xi_{6}=-P_{k/2}p^{k/4}(2(-1)^{k/4}-1)+p^{k/2}$, $\xi_{7}=Q_{k/2}p^{k/4}(2(-1)^{k/4}+1)-5p^{k/2}$ and  $\xi_{8}=-Q_{k/2}p^{k/4}(2(-1)^{k/4}-1)-5p^{k/2}$
\end{theorem}
\begin{proof}
From Theorem \ref{thm3.2} and
(b) in Lemma \ref{lem7}, this theorem follows
immediately.
\end{proof}
\begin{example}
Let $~q=5$ and let $~m=4$. Then the set
$~\mathcal{C}(r,12)$
 is a $[52,4,32]$ code over $~GF(5)$ with the weight distribution
\begin{align*}
 1+52x^{32}+104x^{36}+208x^{40}+104x^{44}+104x^{48}+52x^{52}.
\end{align*}
\end{example}

\begin{theorem}\label{thm5.21.3}
Let $N=12$, $N_{2}=12$ and $p\equiv 7\pmod {12}$. Let $\rho=2(-1)^{k(p+5)/6}$, then
$k\equiv 0\pmod 6$ and
$\mathcal{C}(r,12)$ is a
$[(r-1)/12,m]$  code with the weight distribution
\begin{align*}
1&+\frac{r-1}{12}x^{(q-1)(r-\eta_{1}^{*(12,r)})/12q}
+\cdots\\
&+\frac{r-1}{12}x^{(q-1)(r-\eta_{12}^{*(12,r)})/12q}
,
\end{align*}
where
$\eta_{j}^{*(12,r)}$ is defined as follows:
If $j$ is odd, $\eta_{j}^{*(12,r)}=-(-1)^{k/2}p^{k/6}V_{j,2k/3}
-p^{k/3}V_{2j,k/3}+(-1)^{k/2}p^{k/2}$; if
$2\|j$, $\eta_{j}^{*(12,r)}=-(-1)^{k/2}p^{k/6}V_{j,2k/3}
+p^{k/3}V_{2j,k/3}(\rho-1)+p^{k/2}(\rho-(-1)^{k/2})$; if $4\mid j$, $\eta_{j}^{*(12,r)}=-(-1)^{k/2}p^{k/6}V_{j,2k/3}
-p^{k/3}(\rho+1)V_{2j,k/3}-p^{k/2}(\rho+(-1)^{k/2})$.
\end{theorem}
\begin{proof}
From Theorem \ref{thm3.2} and
(c) in Lemma \ref{lem7}, this theorem follows
immediately.
\end{proof}
\begin{example}
Let $~q=7$ and let $~m=6$. Then the set
$~\mathcal{C}(r,12)$
 is a $[9804,6,8256]$ code over $~GF(7)$ with the weight distribution
\begin{align*}
 1&+9804x^{8256}+9804x^{8280}+9804x^{8340}\\
 &+9804x^{8730}+19608x^{8388}+19608x^{8418}\\
 &+19608x^{8478}+9804x^{8496}+9804x^{8532}.
\end{align*}
\end{example}

\begin{theorem}\label{thm5.21.4}
Let $N=12$, $N_{2}=12$ and $p\equiv 11\pmod {12}$, then
$k\equiv 0\pmod 2$ and
$\mathcal{C}(r,12)$ is a
$[(r-1)/12,m]$  two weight code with the weight distribution
\begin{align*}
1&+\frac{11(r-1)}{12}x^{(q-1)(r-(-1)^{k/2}p^{k/2})/12q}
\\&
+\frac{r-1}{12}x^{(q-1)(r+11(-1)^{k/2}p^{k/2})/12q}.
\end{align*}
\end{theorem}
\begin{proof}
This theorem is
the semi-primitive case in \cite{BM1}.
\end{proof}
\begin{example}
Let $~q=23$ and let $~m=2$. Then the set
$~\mathcal{C}(r,12)$
 is a $[44,2,22]$ code over $~GF(23)$ with the weight distribution
\begin{align*}
 1+44x^{22}+484x^{44}
\end{align*}
\end{example}

\section{Conclusion}
This paper presents necessary and sufficient conditions for an irreducible cyclic code having
only one nonzero weight or the maximal number of distinct nonzero weights. When the number of distinct nonzero weight achieves the maximum, we obtain a divisible property of a codeword.
Further, we determine the weight distributions of $\mathcal{C}(r,N)$ for cases $N_{2}=3,4$ or
$N=5,6,8,12$. From Theorem \ref{thm3.2}, we can immediately obtain the weight distribution of $\mathcal{C}(r,N)$ for cases $N_{2}=5,6,8,12$.


%

\section*{Acknowledgment}

Chungming Tang, Yanfeng Qi and Maozhi Xu
acknowledge support from
the Natural Science Foundation of China
(Grant No.10990011 \& No.60763009). 
Baocheng Wang and Yixian Yang acknowledge support from   
National Natural Science Foundation of China (61003285),
the Fundamental Research Funds for the Central Universities (BUPT2009RC0215) and
the Fundamental Research Funds for the Central Universities (BUPT2011RC0209).

\ifCLASSOPTIONcaptionsoff
  \newpage
\fi


\begin{thebibliography}{99}

\bibitem{AL}  Y. Aubry and P. Langevin, On the weights of binary irreducible
cyclic codes, in Proc. Workshop on Coding and Cryptography, Bergen, Norway, 2005, pp. 161-169.

\bibitem{BM1} L. D. Baumert and R. J. McEliece, Weights of irreducible cyclic codes, Inf. Contr., vol. 20, no. 2, pp. 158-175, 1972.

\bibitem{BM2} L. D. Baumert and J. Mykkeltveit, Weight distributions of some irreducible cyclic codes, DSN Progr. Rep., vol. 16, pp. 128-131, 1973.

\bibitem{DG}P. Delsarte and J. M. Goethals, Irreducible binary cyclic codes of
even dimension, in Proc. 2nd Chapel Hill Conf. Combinatorial Mathematics and Its Applications, Chapel Hill, NC, 1970, pp. 100-113.

\bibitem{Ding2} C. Ding, the Weight Distribution of Some Irreducible Cyclic Codes, IEEE Transactions on Infromation Theory, VOl
55, NO. 3, March 2009.
\bibitem{DH} C. Ding , Niederreiter H., Cyclotomic linear codes of order 3. IEEE Trans. Inform. Theory 53, 2274-2277(2007).


\bibitem{Gu} S. Gurak, Period polynomials for $F_{q}$  of fixed small degree, inNumber
Theory. Providence, RI: Amer. Math. Soc., 2004, pp. 127-145.
\bibitem{HKM} T. Helleseth, T. Klove, and J. Mykkeltveit, The weight distribution of
irreducible cyclic codes with block length $n_{1}((q^l-1)/N)$, Discr.
Math., vol. 18, no. 2, pp. 179-211, 1977.
\bibitem{Ho}  A. Hoshi, Explicit lifts of quintic Jacobi sums and period polynomials
for $F_{q}$, Proc. Japan Acad., vol. 82, no. 7, pp. 87-92, 2006.

\bibitem{La1}S. Lang, Cyclotomic Fields, Graduate Texts in Mathematics, Springer-Verlag: New York, 1978.
\bibitem{La2}S. Lang, Cyclotomic Fields,  II. Graduate Texts in Mathematics, Springer-Verlag: New York, 1980.


\bibitem{La} P. Langevin, A new class of two weight codes, in Finite Fields and
Applications, S. Cohen and H. Niederreiter, Eds. Cambridge, U.K.: Cambridge Univ. Press, 1996, pp. 181-187.

\bibitem{LH} R. Lidl and H. Niederreiter, Introductrion to Finite Fields and Their applications, Cambridge University Press, Cambridge, UK, 1994.




\bibitem{MS} F. MacWilliams and J. Seery, The weight distributions of some minimal cyclic codes, IEEE Trans. Inf. Theory, vol. IT-27, no. 6, pp.
796-806, Nov. 1981.
\bibitem{Mc1} R. J. McEliece, A class of two-weight codes, Jet Propulsion Laboratory Space Program Summary 37¨C41, vol. IV, pp. 264-266.
\bibitem{Mc2} R. J. McEliece, Irreducible cyclic codes and Gauss sums, in Combinatorics, Part 1: Theory of Designs, Finite Geometry and Coding
Theory. Amsterdam, TheNetherlands:Math. Centrum, 1974, vol. 55,
Math. Centre Tracts, pp. 179-196.
\bibitem{MR} R. J. McEliece and H. Rumsey Jr., Euler products, cyclotomy, and
coding, J. Number Theory, vol. 4, pp. 302-311, 1972.

\bibitem{MV} M. J. Moisio and K. O. Vaanen, ¡°Two recursive algorithms for com-
puting theweight distribution of certain irreducible cyclic codes,¡± IEEE
Trans. Inf. Theory, vol. 45, no. 3, pp. 1244-1249, May 1999.

\bibitem{My} G. Myerson, Period polynomials and Gauss sums for finite fields,
Acta Arith., vol. 39, pp. 251-264, 1981.


\bibitem{SW} B. Schmidt and C. White, All two-weight irreducible cyclic codes,
Finite Fields Appl., vol. 8, pp. 1-17, 2002.


\bibitem{St} L. Stickelberger, $\ddot{U}$ber eine Verallgemeinerung der Kreistheilung, Math. Annal. 37 (1890),321-367.

\bibitem{VV} M. van der Vlugt, On the weight hierarchy of irreducible cyclic
codes, J. Comb. Theory Ser. A, vol. 71, no. 1, pp. 159-167, Jul. 1995.


\bibitem{Wa} L. C. Washington, Introduction to Cyclotomic Fields, Graduate Texts in Math., Vol. 83. Springer-Verlag, Berlin/New York/Heidelberg,
1997.


\end{thebibliography}
\end{document}